\documentclass{article}

\title{An easier way to compute 2-cocycles coming from a reduction for semidirect products}
\author{Viacheslav Goncharov\footnote{IST Austria -- Institute of Science and Technology Austria, 3400 Klosterneuburg, viacheslav.goncharov@ist.ac.at}}

\usepackage{natbib}

\usepackage{graphicx} 
\usepackage{amsmath}
\usepackage[utf8]{inputenc}
\usepackage[makeroom]{cancel} 
\usepackage{amsthm}
\usepackage{amssymb}
\usepackage{yfonts}
\usepackage{blindtext}
\usepackage{pst-node}
\usepackage{xcolor} 
\usepackage{leftindex}
\usepackage{tikz-cd} 
\usepackage{stmaryrd} 
\usepackage[T1]{fontenc}
\usepackage{comment}
\usepackage{geometry}
\newcommand\restr[2]{{
  \left.\kern-\nulldelimiterspace 
  #1 
  \littletaller 
  \right|_{#2} 
  }}

\newcommand{\R}{\mathbb{R}}
\newcommand{\g}{\mathfrak{g}}
\newcommand{\h}{\mathfrak{h}}

\newcommand{\reallywidehat}[1]{%
\savestack{\tmpbox}{\stretchto{%
  \scaleto{%
    \scalerel*[\widthof{\ensuremath{#1}}]{\kern-.6pt\bigwedge\kern-.6pt}%
    {\rule[-\textheight/2]{1ex}{\textheight}}
  }{\textheight}%
}{0.5ex}}%
\stackon[1pt]{#1}{\tmpbox}%
}

\newtheorem{theorem}{Theorem}
\newtheorem{lemma}{Lemma}

\newtheorem{proposition}{Proposition}

\numberwithin{equation}{section}

\begin{document}

\maketitle

\begin{abstract}
    For Hamiltonian actions of semidirect products $G=F \ltimes H$, we study 2-cocycles arising from residual Hamiltonian actions of $F$ on Hamiltonian reductions for $H$. The motivation comes from the study of Teichmüller spaces for surfaces with boundary, which carry Hamiltonian actions of the Virasoro algebra. In this paper, we give a general setup for the problem, and we suggest an easier way to obtain the Gelfand-Fuchs 2-cocycles for Hamiltonian actions on Teichmüller spaces.
\end{abstract}

    

\tableofcontents



\newpage


\section{Intoduction}

Recall that a Hamiltonian action of a connected Lie group $G$ on a symplectic manifold $\mathcal{M}$ induces a (trivial or nontrivial) central extension of the corresponding Lie algebra $\mathfrak{g}={\rm Lie}(G)$.  In this paper, we study Hamioltonian actions of semi-direct products $G = F \ltimes H$ on symplectic manifolds. We denote the corresponding Lie algebras by $\mathfrak{f}={\rm Lie}(F)$ and $\mathfrak{h}={\rm Lie}(H)$, respectively. If the momentum maps are not equivariant, the Lie algebras $\mathfrak{g}$ and $\mathfrak{h}$ may carry central extensions induced by the Hamiltonian action. 

For more details on reduction under the action of semi-direct products and of the corresponding central extensions, see \cite{SympRedForSDP}, \cite{CzechJofPh}.

For a $G$-Hamiltonian space $\mathcal{M}$, we consider the reduced space $\mathcal{M}_{\rm red}=\mathcal{M}//_{(\eta, a)}H = (\mathcal{M}\times\mathcal{O}^G_{(\eta, a)})\sslash_0 H$ under the $H$-action (the construction is also known as shifting trick). Here $a\in \mathbb{R}$ is the level of the central extension,  $\eta \in \mathfrak{g}^*$, and by abuse of notation we denote by the same letter its restriction to $\mathfrak{h}$. Furthermore, we denote by $\mathfrak{g}_{(\eta, a)}$ the stabiliser of $(\eta, a)$ under the (affine) coadjoint action of $\mathfrak{g}$.

Under an important (and restrictive) assumption
\begin{equation}
    \mathfrak{g} = \mathfrak{g}_{(\eta, a)} \oplus \mathfrak{h},
\end{equation}
the reduced space $\mathcal{M}_{\rm red}$ carries a Hamiltonian action of $F$, and this action may induce a non-trivial central extension of $\mathfrak{f}$. Our first main result (Theorem \ref{MainTh}) is that this central extension can be easily computed by restricting the central extension of $\mathfrak{g}$ induced by the Hamiltonian action on $\mathcal{M}$ to $\mathfrak{g}_{(\eta, a)} \cong \mathfrak{h}$. This allows one to avoid direct computation for an action on $\mathcal{M}_{\rm red}$.

Motivation for this work comes from the study of  infinite-dimensional  Teichmüller spaces associated to an oriented surfaces with ideal boundary (see \cite{AlArt}). These spaces may be obtained by Hamiltonian reduction under the loop group $H=LB$, where $B \leq {\rm PSL}(2, \mathbb{R})$ is the Borel subgroup. In this case, the group $F={\rm Diff}^+(S^1)$ is the group of orientation preserving diffeomorphisms of the circle, and the Teichmüller space carries a Hamiltonian action of the Virasoro algebra. Our second main result is a new computation of this central extension, which turns out to be easier than the original computation in \cite{AlArt}. For more information on infinite-dimensional groups and their coadjoint orbits, see \cite{SomeAppPh}, \cite{GeomInfDimGr}, \cite{Kirillov}, \cite{Witten}, \cite{CzechJofPh}.

The structure of the paper is as follows.
In Section 2,  we briefly recall the necessary background material and compute the central extension of $\mathfrak{f}$ corresponding to the Hamiltonian action of $F$ on $\mathcal{M}_{\rm red}$.
In Section 3, we set $G ={\rm Diff}^+(S^1) \ltimes  LB$, and we consider three examples of Hamiltonian $G$-actions. The first family of examples are coadjoint  orbits of the canonical central extension of the loop group $L({\rm PSL}(2, \mathbb{R}))$. In this case, our procedure gives rise to a version of the Drinfeld-Sokolov reduction \cite{DrinSok}. In the second example, we consider moduli spaces of flat ${\rm PSL}(2, \mathbb{R})$-connections on surfaces with boundary, and in the third example infinite-dimensional Teichmüller spaces associated to surfaces with ideal boundaries (see \cite{AlArt}). The last example is the main motivation for this paper. In particular, it justifies the choice for the group $G ={\rm Diff}^+(S^1) \ltimes  LB$.




\paragraph{Acknowledgments.} This work is based on my Master thesis at the University of Geneva. I am grateful to my supervisor Anton Alekseev for setting up the problem, and for his guidance. I would like to thank Damien Calaque for useful discussions, and Rea Dalipi for helpful remarks on the text of the article.

\section{Introducing the problem}\label{Sec2}


\subsection{Introductory definitions and the setup}

Here we recall necessary definitions for our problem. The material is classical and we refer to \cite{GeomInfDimGr}, \cite{SomeAppPh}, \cite{KirOrbMeth} for the details. 

Let $G$ be a Lie group and $\g$ its Lie algebra (in our examples they will be infinite-dimensional). Suppose there is a Hamiltonian action of $G$ on a symplectic manifold $(\mathcal{M}, \omega)$ (this will also be infinite-dimensional), i.e. there is a momentum map

\begin{equation}
        \mu\colon\mathcal{M}\longrightarrow\g^*
\end{equation}
such that $\iota_x\omega_m = d\langle\mu(m), x\rangle$ for any $x\in\g$, where $m\in\mathcal{M}$.

There is coadjoint action of the group $G$ on $\g^*$. The momentum map might not be equivariant. Non-equivariancy is measured by a 2-cocycle on the Lie algebra. The 2-cocycle is defined in the following way: 

\begin{equation}
    c(x, y) = \{\langle\mu, x\rangle,\langle\mu, y\rangle\} - \langle\mu, [x, y]\rangle,
\end{equation}
where $\{\cdot, \cdot\}$ are Poisson brackets and $x, y\in\g$. However, for our purposes, it is more convenient to use the following equivalent formula for the 2-cocycle: 

\begin{equation}\label{DefOfCocycle}
    c(x, y) = \mathcal{L}_x\langle\mu, y\rangle - \langle\mu, [x, y]\rangle,
\end{equation}
where $\mathcal{L}_x$ is the Lie derivative. A priori, $c(x, y)$ is a function on $\mathcal{M}$ but one can show that it is, in fact, a constant. 
Such 2-cocycle defines a central extension $\widehat{\g} = \g\oplus\R$ of the Lie algebra $\g$, i.e. there is a short exact sequence of Lie algebras

\begin{equation}
0\longrightarrow\R\longrightarrow\widehat{\g}\longrightarrow\g\longrightarrow 0,
\end{equation}
and a bracket on $\widehat{\g}$ is given by

\begin{equation}
    [(x, t_1), (y, t_2)]_{\widehat{\g}} = ([x, y]_{\g}, c(x, y)).
\end{equation}

We recall that central extensions of a Lie algebra are classified by its second cohomology group with coefficients in the trivial module $\R$. Two 2-cocycles $c_1, c_2\in\Lambda^2\g^*$ represent the same cohomology class if and only if $c_1 - c_2 = d\beta$, where $\beta\in\g^*$ and $d\beta\in\Lambda^2\g^*$ is defined by $d\beta(x, y) = \beta([x, y])$.

Now suppose that the group $G$ is a semidirect product of groups $H\trianglelefteq G$ and $F = G/H$, i.e. there is a short exact sequence of Lie groups with a section $s\colon F\rightarrow G$: 

\begin{center}
    \begin{tikzcd}
        1\arrow[r] & H \arrow[r] & G\arrow[r, yshift=-0.5ex] & F \arrow[l, yshift=0.5ex, "s" description]\arrow[r] & 1
    \end{tikzcd}    
\end{center}
In other words, $G = HF$ and $H\cap F = \{e\}$. We will also write $G = F\ltimes H$. We denote by $\h\trianglelefteq\g$ and $\mathfrak{f} = \g/\h$ Lie algebras corresponding to $H$ and $F$.

\subsection{Description of the problem}

A dual space $\widehat{\g}^*$ to the centrally extended Lie algebra $\widehat{\g}$ is given by $\widehat{\g}^* = \{(\eta, a)\ \ \vert \ \ \eta\in\g^*, a\in\R\}$ with a pairing $\langle(\eta, a), (x, t)\rangle = \langle\eta, x\rangle + at$. The group $G$ acts on $\widehat{\g}^*$ by coadjoint action: $\text{Ad}^*_g(\eta, a) = (\text{Ad}^*_g\eta, a)$ (strictly speaking, one should consider an action of a central extension $\widehat{G}$ of the group $G$, but the centre acts trivially, so we omit writing it). An infinitesimal version of this action is the coadjoint action of the corresponding Lie algebra: $\langle\text{ad}^*_y (\eta, a), (x, t)\rangle = \langle(\eta, a), [(x, t), (y, 0)]_{\widehat{\g}}\rangle$.

Take a functional $(\eta, a)\in\widehat{\g}^*$. Coadjoint orbit $\mathcal{O}_{(\eta, a)}^G\subseteq\widehat{\g}^*$ of the functional is naturally a symplectic manifold with the standard KKS symplectic form (for details see \cite{KirReps}):

\begin{equation}
    \omega^{KKS}_{(\eta, a)}(\text{ad*}_{(x, t_1)} (\eta, a), \text{ad*}_{(y, t_2)} (\eta, a)) = \langle(\eta, a), [(x, t_1), (y, t_2)]_{\widehat{\g}}\rangle,
\end{equation}

An action of the group $G$ on the orbit is Hamiltonial with a momentum map given by an inclusion $i\colon \mathcal{O}_{(\eta, a)}^G\hookrightarrow\widehat{\g}^*$. Thus, we can consider a reduced space $\mathcal{M}_{\eta}$ defined by

\begin{equation}
    \mathcal{M}_{\eta} = (\mathcal{M}\times\mathcal{O}_{\eta}^G)\sslash_0 H,
\end{equation}
where $H$ acts diagonally on $\mathcal{M}\times\mathcal{O}_{\eta}^G$ (this procedure is also known as "shifting trick"). 

There is a residual action of $F = G/H$ on $\mathcal{M}_{\eta}$. This is a general fact: if $G$ acts on $X$, then for $H\trianglelefteq G$ there is an action $G/H\curvearrowright X/H$. 

This action is still Hamiltonian. Indeed, consider the following diagram: 

\begin{center}
        \begin{tikzcd}
            G\curvearrowright & \mathcal{M}\times\mathcal{O}^G_{\eta}\arrow[r, "\widetilde{\mu}"] & \g^*\arrow[d, equal]\\
            G\curvearrowright & \widetilde{\mu}^{-1}(0)\arrow[d, two heads, "p"]\arrow[u,"i", hook]\arrow[r, "\widetilde{\mu}"] & \g^*\\
            G/H\curvearrowright & \mathcal{M}_{\eta}\arrow[r, "\mu_{res}"] & (\g/\h)^*\arrow[u, "pr^*", tail]
        \end{tikzcd}
    \end{center}
We need to show that $\mu_{res}$ is a momentum map. We are given the following: 
\begin{enumerate}
    \item $(\mathcal{M}\times\mathcal{O}^G_{\eta}, \widetilde{\omega})$ and $(\mathcal{M}_{\eta}, \omega_{red})$ are symplectic manifolds and $p^*\omega_{red} = i^*\widetilde{\omega}$.
    \item $\widetilde{\mu}$ is a momentum map, i.e. $d\langle\widetilde{\mu}, x\rangle = \iota_x\widetilde{\omega}$ for $x\in\g$.
    \item $i^*\widetilde{\mu} = pr^*\circ\mu_{res}\circ p$
\end{enumerate}
We have $d\langle pr^*\circ\mu_{res}\circ p, x\rangle = \iota_x i^*\widetilde{\omega} = p^*\iota_{pr(x)}\omega_{red}$. Since $p^*$ is injective we obtain that $\mu_{res}$ is a momentum map. This momentum map might still have a nontrivial 2-cocycle  $c_{res}\in\Lambda^2(\g/\h)^*$. It turns out that in concrete examples $c_{res}$ is hard to compute due to the complexity of the construction of $\mathcal{M}_{\eta}$. The problem is to simplify the way to obtain this 2-cocycle.

\subsection{Special condition}


\subsubsection{Intermediate results}

We claim that if we pick up a functional $\eta\in\g^*$ such that 

\begin{equation}\label{TheCond}
    \g_{(\eta, a)} \oplus \h = \g,
\end{equation} 
then the desired 2-cocycle is the same as for an action of $G_{(\eta, a)}\leq G$. Here $\g_{(\eta, a)}\leq\g$ is a stabiliser of $(\eta, a)$ under the coadjoint action of Lie algebra $\g$. In Section \ref{Sec3} we give an example of such functional. First, let us note one consequence of this condition. 

\begin{lemma}
    Suppose that the groups $G, H, F$ and $G_{(\eta, a)}$ are connected and simply connected, and the functional $\eta\in\g^*$ satisfies (\ref{TheCond}). Then $G = G_{(\eta, a)}H = HG_{(\eta, a)}$.
\end{lemma}

\begin{proof}
    Recall that $G = F\ltimes H$.  
    One of the definitions of a semidirect product is that the following composition of maps

    \begin{center}
        \begin{tikzcd}
            F\arrow[r, hook, "i"] & G \arrow[r, "\pi"] & G/H
        \end{tikzcd}
    \end{center}

    \noindent
    is an isomorphism between $F$ and $G/H$. In particular, we have $\mathfrak{f}\oplus\h = \g$. Due to condition (\ref{TheCond}) we obtain an isomorphism $\mathfrak{f}\cong\g_{(\eta, a)}$. Since $F$ and $G_{(\eta, a)}$ are connected and simply connected, we have an isomorphism $G_{(\eta, a)}\cong F$. Precomposing it with the inclusion $i$ above, we obtain $G = G_{(\eta, a)}\ltimes H$, i.e. $G = G_{(\eta, a)} H$ and $G_{(\eta, a)}\cap H = \{e\}$. 
\end{proof}
Having an action of the group $G$, one can restrict to an action of the subgroup $H$. In the case of the coadjoint action on $\widehat{\g}^*$, we then obtain two orbits $\mathcal{O}^G_{(\eta, a)}$ and $\mathcal{O}^H_{(\eta, a)}$.
An immediate consequence of the lemma is that the orbits $\mathcal{O}^G_{(\eta, a)}$ and $\mathcal{O}^H_{(\eta, a)}$ coincide: 

\begin{equation}
    \mathcal{O}^G_{(\eta, a)} = \text{Ad*}_G (\eta, a)= \text{Ad*}_{HG_(\eta, a) } (\eta, a) = \text{Ad*}_{H} \text{Ad*}_{G_(\eta, a) } (\eta, a) = \text{Ad*}_H (\eta, a) = \mathcal{O}^H_{(\eta, a)}
\end{equation}
But we can make this statement more precise.

\begin{proposition}\label{SympOfOrbs}
    Let $(\eta, a)$ be a functional from $\widehat{\g}^*_s$ subject to condition (\ref{TheCond}). Suppose that central extensions $\widehat{\h}$ and $\widehat{\g}$ are given by such 2-cocycles $\gamma_{\h}$ and $\gamma_{\g}$ that $\gamma_{\g}\vert_{\h\times\h} = \gamma_{\h}$.
    Then coadjoint orbits $\mathcal{O}^G_{(\eta, a)}$ and $\mathcal{O}^H_{(\eta, a)}$ are symplectomorphic.
\end{proposition}

\begin{proof}

Consider the following diagram: 

\begin{center}
    \begin{tikzcd}
        G\arrow[r, two heads]\arrow[d, two heads] & \mathcal{O}^G_{(\eta, a)}\arrow[d, equal]\\
        H\arrow[r, two heads] & \mathcal{O}^H_{(\eta, a)}
    \end{tikzcd}
\end{center}
Horizontal maps are given by $g\mapsto\text{Ad}^*_g(\eta, a)$ and $h\mapsto\text{Ad}^*_h(\eta, a)$. Strictly speaking, one should consider two functionals: one on $\g$ and another on $\h$. But in our case the second one is just a restriction of the first one, so we do not change the notation. 
The corresponding diagram for tangent spaces is 

\begin{center}
    \begin{tikzcd}
        \h\oplus\mathfrak{p} = \g\arrow[r, two heads]\arrow[d, "\Pi", two heads] & (T_{(\eta, a)}\mathcal{O}^G_{(\eta, a)}, \omega^G_{(\eta, a)})\arrow[d]\\
        \h\arrow[r, two heads] & (T_{(\eta, a)}\mathcal{O}^H_{(\eta, a)}, \omega^H_{(\eta, a)}),
    \end{tikzcd}
\end{center}
where $\mathfrak{p} = \ker\Pi$. Thus 

\begin{equation}
    T_{(\eta, a)}\mathcal{O}^G_{(\eta, a)} = \g/\g_{(\eta, a)} = \h = T_{(\eta, a)}\mathcal{O}^H_{(\eta, a)},
\end{equation}
where we regard $\g/\g_{(\eta, a)}$ and $\h$ as vector spaces. By taking $x + s_x , y + s_y \in\h \oplus \g_{(\eta, a)} = \g$, we verify the symplectomophism

\begin{equation}
    \begin{split}
        \omega^G_{(\eta, a)}(\text{ad*}_{x + s_x}(\eta, a), \text{ad*}_{y + s_y}(\eta, a)) &= \omega^G_{(\eta, a)}(\text{ad*}_{x}(\eta, a), \text{ad*}_{y}(\eta, a)) = \langle(\eta, a), [x, y]\rangle_{\g} = \\
        &= \langle(\eta, a), [x, y]\rangle_{\h} = \omega^H_{(\eta, a)}(\text{ad*}_{x}(\eta, a), \text{ad*}_{y}(\eta, a)) = \\
        &= \omega^H_{(\eta, a)}(\text{ad*}_{\Pi(x + s_x)}(\eta, a), \text{ad*}_{\Pi(y + s_y)}(\eta, a))
    \end{split}
\end{equation}

Here we used the condition on the 2-cocycles when we moved from the first line to the second one. Indeed, let $x, y \in\h$. Then

\begin{equation}
    \langle(\eta, a), [x, y]\rangle_{\g} = \langle\eta, [x, y]\rangle + a\gamma_{\g}(x, y) = \langle\eta, [x, y]\rangle + a\gamma_{\h}(x, y) = \langle(\eta, a), [x, y]\rangle_{\h}
\end{equation}

\end{proof}

Now we can formulate and prove the main result. 


\subsubsection{Main theorem}

\begin{theorem}\label{MainTh}
    \begin{itemize}

        \item Let $G$ be a semidirect product $F\ltimes H$ of connected and simply connected Lie groups $H\trianglelefteq G$ and $F$ with Lie($G$)$=\g$ and Lie($H$)$=\h$, and $\widehat{\g}, \widehat{\h}$ -- central extensions. 

        \item Let $a\in\R$ and $\eta\in\g^*$ be a functional on the Lie algebra $\g$ such that $\g_{(\eta, a)} \oplus \h = \g$.

        \item Let $\mathcal{M}$ be a symplectic manifold with a Hamiltonian action of $G$ with a momentum map $\mu$ and a corresponding 2-cocycle $c\in\Lambda^2\g^*$.

        \item Let $\mathcal{M}' = (\mathcal{M}\times\mathcal{O}^G_{(\eta, a)})\sslash_0H$ be a reduced space under the diagonal action of $H$.
    \end{itemize}

    \vspace{8pt}
        Then the claim is that under the isomorphism $\g/\h\cong\g_{(\eta, a)}$ the 2-cocycle $c_{res}$ coming from a residual action $G/H\curvearrowright\mathcal{M}'$ coincides with the 2-cocycle $c_{stab}$ coming from a restriction of the action $G\curvearrowright\mathcal{M}$ to an action of the stabiliser $G_{(\eta, a)}$. 
    
\end{theorem}

\begin{proof}
    
    First, let us note that 

    \begin{equation}
        \begin{split}
            \mathcal{M}' &= (\mathcal{M}\times\mathcal{O}^G_{(\eta, a)})\sslash_0H = (\mathcal{M}\times\mathcal{O}^H_{(\eta, a)})\sslash_0H = \\
            &=\{(m, \psi)\in\mathcal{M}\times\mathcal{O}^H_{(\eta, a)}\big| \mu_H(m) - \psi = 0\}/H = \\
            &= \{(h.m, \eta)\big|h\in H, \mu_H(h.m) = \eta\} = \mu_H^{-1}(\eta)\subseteq\mathcal{M},
        \end{split}
    \end{equation}
    where $\mu_H\colon\mathcal{M}\rightarrow\h^*$ is a momentum map for an $H$-action. We used Lemma \ref{SympOfOrbs} in the beginning when we replaced $\mathcal{O}^G_{(\eta, a)}$ by $\mathcal{O}^H_{(\eta, a)}$. Thus, due to the condition $\g_{(\eta, a)} \oplus \h = \g$ the map

    \begin{center}
        \begin{tikzcd}
            (\g/\h)^*\arrow[r, "pr*", tail] & \g^*\arrow[r, "j^*", two heads] & \g^*_{(\eta, a)}
        \end{tikzcd}
    \end{center}
    is an isomorphism. Now, look at the following diagram: 

    \begin{center}
        \begin{tikzcd}
            G_{(\eta, a)}\curvearrowright & \mathcal{M}\arrow[r, "\mu_{(\eta, a)}"]\arrow[d, "id", equal] & \g^*_{(\eta, a)}\\
            G\curvearrowright & \mathcal{M}\arrow[r, "\mu"] & \g^*\arrow[u, "j^*", two heads]\\
            G/H\cong G_{(\eta, a)}\curvearrowright & \widetilde{\mathcal{M}} = \mu_H^{-1}(\eta)\arrow[r, "\overline{\mu}"] \arrow[u,"i", hook]& (\g/\h)^*\arrow[u, "pr^*", tail]
        \end{tikzcd},
    \end{center}
    where all three actions are Hamiltonian. 
    The two-cocycle $c_{res}$ coming from the bottom line is defined by 

    \begin{equation}
        c_{res}(x, y) = \mathcal{L}_x\langle\overline{\mu}(m), y\rangle - \langle\overline{\mu}(m), [x,y]\rangle
    \end{equation}
    for any $x, y\in\g/\h$. And the two-cocycle $c_{stab}$ coming from the top line is defined by

    \begin{equation}
        c_{stab}(x, y) = \mathcal{L}_x\langle\mu_{(\eta, a)}(m), y\rangle- \langle\mu_{(\eta, a)}(m), [x,y]\rangle
    \end{equation}

    Since the composition $j^*\circ pr^*$ is an isomophism, the momentum maps $\mu_{(\eta, a)}$ and $\overline{\mu}$ are the same on $\mu_H^{-1}(\eta)\subseteq\mathcal{M}$, so two-cocycles $c_{res}$ and $c_{stab}$ also coincide there. But these 2-cocycles are always constants as functions on $\mathcal{M}$. This finishes the proof of the theorem. 
\end{proof}

\section{Examples}\label{Sec3}



\paragraph{The group and the algebra.}

In our examples $G$ will be a semidirect product of two infinite-dimensional groups.    
The first group is a universal cover of diffeomorphims of a circle preserving orientation: 

\begin{equation}
    \widetilde{\text{Diff}^+}(S^1) := \{\varphi\in C^{\infty}(\R)\big| \varphi(\theta + 2\pi) = \varphi(\theta) + 2\pi, \varphi'>0\},
\end{equation}
where the universal cover is given by

\begin{equation}\label{DiffCover}
    \begin{split}
        \widetilde{\text{Diff}^+}(S^1)&\longrightarrow\text{Diff}^+(S^1)\\
        \varphi&\longmapsto (e^{i\theta}\mapsto e^{i\varphi(\theta)})
    \end{split}
\end{equation}
Product in the group is defined by composition. 

The corresponding Lie algebra is the algebra of smooth vector fields on a circle: $\mathfrak{X}(S^1):= \{f\partial_{\theta}\big| f\in C^{\infty}(S^1)\}$ with the bracket

\begin{equation}
    [f\partial_{\theta}, g\partial_{\theta}] = (fg' - f'g)\partial_{\theta}
\end{equation}

The second group is the loop group of Borel subgroup of $SL_2(\R)$, i.e. the set 

$$LB := \text{Maps}\left(S^1,\left\{ \left(\begin{array}{cc}
    X & Y \\
    0 & X^{-1}
\end{array}\right)\Big| X\in\R_{>0}, Y\in\R\right\}\right)$$
with pointwise multiplication. The corresponding Lie algebra is given by 

$$Lb := \text{Maps}\left(S^1, \left\{\left(\begin{array}{cc}
    x & y \\
    0 & -x
\end{array}\right)\Big| x, y\in\R\right\}\right)$$
with a pointwise bracket. 

\vspace{10pt}
$\text{Diff}^+(S^1)$ acts on $LB$ by precomposition and we lift this action to $\widetilde{\text{Diff}^+}(S^1)$. This gives us the map 

\begin{equation}
    \begin{split}
        \widetilde{\text{Diff}^+}(S^1)\longrightarrow\text{Diff}^+(S^1)&\longrightarrow\text{Aut}(Lb)\\
        \sigma&\longmapsto (b\mapsto b(\sigma))
    \end{split}
\end{equation}
Of course, there is a corresponding map for Lie algebras:

\begin{equation}
    \begin{split}
        \mathfrak{X}(S^1)&\longrightarrow\text{Der}(Lb)\\
        f\partial_{\theta}&\longmapsto(b\mapsto fb')
    \end{split}
\end{equation}
Having this action, one can define semidirect product of these groups. Here we give a definition of a central extenstion of this product and its Lie algebra right away: 

\begin{equation}
    \begin{split}
        \widetilde{\text{Diff}}^+(S^1) &\ltimes L B\oplus\R \\
        \mathfrak{X}(S^1)&\ltimes L b\oplus\R
    \end{split}
\end{equation}

\noindent
with multiplication and bracket given by 

\begin{equation}\label{DefOfMultBr}
    \begin{split}
        (\varphi, B, \delta_1)\cdot(\psi, D, \delta_2) &= (\varphi\circ\psi, B(\psi)D, C((\varphi, B), (\psi, D)) \\
        [(f\partial, b, \delta_1), (g\partial, d, \delta_2)] &= ([f\partial, g\partial], [d, b] + fd' - gb',  c((f\partial, b), (g\partial, d)),
    \end{split}
\end{equation}
where $C(\cdot, \cdot)$ is a group 2-cocycle and $c(\cdot, \cdot)$, is a Lie algebra 2-cocycle (we comment on the last one below). 


\paragraph{The functional.}

Let us show that a functional $\eta\in\g^*$ satisfying (\ref{TheCond}) is given by

\begin{equation}\label{FuncEx}
    \langle(\eta, a),(g\partial, b, t\rangle= \left\langle(\eta, a),\left(g\partial, \left( \begin{array}{cc}
        x &  y\\
        0 & -x
    \end{array} \right), t\right)\right\rangle = \int_{S^1}yd\theta + at
\end{equation}

One might compute that for the Lie algebra $\g = \mathfrak{X}(S^1)\ltimes L b$ the vector space $H^2(\g, \R)$ is three-dimensional. A basis for this space might be chosen to be $[c_1], [c_2], [c_3]$, where

\begin{equation}
    \begin{split}
        c_1((f_1\partial, b_1), (f_2\partial, b_2)) &= \int_{S^1}x_1x_2'd\theta \\
        c_2((f_1\partial, b_1), (f_2\partial, b_2)) &= \int_{S^1}(f_1''x_2 - f_2''x_1)d\theta \\
        c_3((f_1\partial, b_1), (f_2\partial, b_2)) &= \int_{S^1}(f_1''f_2'- f_1'f_2'')d\theta 
    \end{split}
\end{equation}
The last 2-cocycle $c_3$ is called the Gelfand-Fuchs cocycle and will play the main role in the sequel.
Thus, the Lie algebra 2-cocycle $c$ in (\ref{DefOfMultBr}) is some linear combination $c = \lambda_1 c_1 + \lambda_2 c_2 + \lambda_3 c_3$.

\paragraph{Remark.} The algebra $\g = \mathfrak{X}(S^1)\ltimes L b$ is similar to the algebra $\mathfrak{X}(S^1)\ltimes C^{\infty}(S^1)$ in sence that they have the same 2-cocycles. The latter algebra is studied, for example, in \cite{CzechJofPh} and \cite{Marshall}. 

Let us compute a stabiliser $\g_{(\eta, a)}$ explicitly. Consider a coadjoint action of an element $(f\partial, d) = \left(f\partial, \begin{pmatrix}
    z & t\\
    0 & -z
\end{pmatrix}\right)\in\g$ on our functional: 

\begin{equation}\label{CalcStab}
    \begin{split}
        \langle\text{ad}^*_{(f\partial, d)}&(\eta, a), (g, b)\rangle = \left\langle(\eta, a), \left[\left(g\partial, \begin{pmatrix}
            x & y\\
            0 & -x
        \end{pmatrix}\right), \left(f\partial, \begin{pmatrix}
            z & t\\
            0 & -z
            \end{pmatrix}\right)\right]_{\widehat{\g}}\right\rangle = \\
        &= \int_{S^1}2(zy - tx)d\theta + \int_{S^1}(gt' - fy')d\theta + a\lambda_1c_1(b, d) + a\lambda_2 c_2((g\partial, b), (f\partial, d)) + a\lambda_3 c_3(g\partial, f\partial) = \\
        &= \int_{S^1}y(2z + f')d\theta + \int_{S^1}x(-2t + a\lambda_1 z' - a\lambda_2 f'')d\theta + \int_{S^1}g(t' + a\lambda_2z'' + 2a\lambda_3 f''')d\theta
    \end{split}
\end{equation}
An element $(f\partial, d)$ belongs to a stabiliser $\g_{(\eta, a)}$ if and only for any $ x, y, g\in C^{\infty}(S^1)$ the expression above is zero. This is equivalent to the following system:

\begin{equation}
    \begin{cases}
        z = -\frac{1}{2}f'\\
        2t = -a(\lambda_1/2 + \lambda_2)f''\\
        t' = -a(-\lambda_2/2 + 2\lambda_3)f'''
    \end{cases}
\end{equation}
For these equations to be consistent ($t'$ can be expressed in two different ways), the following equality must hold: 

\begin{equation}\label{RelationForLambdas}
    \frac{\lambda_1}{4} + \lambda_2 - 2\lambda_3  = 0
\end{equation}
Then, our stabiliser is 

\begin{equation}\label{LAstab}
    \begin{split}
        \g_{(\eta, a)} &= \left\{\left(f\partial,\frac{1}{2}\left( \begin{array}{cc}
            -f' &  -a(\lambda_1/2 + \lambda_2)f''\\
            0 & f'
        \end{array} \right)\right)\bigg|    f\partial\in\mathfrak{X}(S^1)\right\} = \\
    &=\{(f\partial, N_f)\big| f\partial\in\mathfrak{X}(S^1)\}
    \end{split}
\end{equation}

\begin{proposition}
    The functional $(\eta, a)$ defined by (\ref{FuncEx}) satisfies the condition $\g_{(\eta, a)}\oplus\h = \g$ for $\h = Lb\trianglelefteq \g$.
\end{proposition}

\begin{proof}
    Indeed, the first component of an element in $\g_{(\eta, a)}$ can be an arbitrary vector field $f\partial$ on the circle $S^1$, and when $f\equiv0$, the second component $N_f$ also vanishes. 
    In its turn, $\h$ does not contain a vector field component, so $\g_{(\eta, a)}\cap\h = \{0\}$. 
\end{proof}

We note that due to condition (\ref{RelationForLambdas}) we always have $c((f\partial, N_f), (g\partial, N_g)) = 0$ for any two elements $(f\partial, N_f), (g\partial, N_g)\in\g_{(\eta, a)}$.

\subsection{The first example: coadjoint orbit}

Let $\bar{G} = \widetilde{SL_2}(\R)$ and $\bar{\g} =$ Lie$(\bar{G}) = \mathfrak{sl}_2(\R)$. We take a functional $\ell := (A, 1)\in \widehat{L\bar{g}}^*$ and consider its coadjoint orbit $\mathcal{O}_{\ell}^{L\bar{G}}$. This orbit is our $\mathcal{M}$ from the theorem. As a group $G$ we take $\widetilde{\text{Diff}^+}(S^1)\ltimes LB$. 

\noindent
Since $G = \widetilde{\text{Diff}^+}(S^1)\ltimes LB \leq \widetilde{\text{Diff}^+}(S^1)\ltimes L\bar{G}$ there is an action $G\curvearrowright\mathcal{O}_{\ell}^{L\bar{G}}$: an action of $\varphi\in\widetilde{\text{Diff}^+}(S^1)$ is realised by covering (\ref{DiffCover}) and then an action of an image of $\varphi$ in ${\text{Diff}^+}(S^1)$.

\noindent
To compute a 2-cocycle for this action, we need to know a momentum map. Let us proceed step by step. 

\subsubsection{Loop group}

For the action $LB\curvearrowright\mathcal{O}_{\ell}^{L\bar{G}}$ the momentum map is known and it is given by just taking the lower part of $A\in L\bar{\g}$:

    \begin{equation}\label{LoopMM}
        \begin{split}
            \mu_{loop}\colon \mathcal{O}_{\ell}^{L\bar{G}}&\longrightarrow Lb^*\cong Lb^{-}\\
            (A, 1)&\longmapsto A^{-}
        \end{split},
    \end{equation}
    where 

    \begin{equation*}
        A = \begin{pmatrix}
                \frac{1}{2}x^* & z^* \\
                y^* & -\frac{1}{2}x^*
            \end{pmatrix} \text{ and } A^{-} = \begin{pmatrix}
                \frac{1}{2}x^* & 0 \\
                y^* & -\frac{1}{2}x^*
            \end{pmatrix}\textbf{}
    \end{equation*}
    We will also need $A^{+} = A - A^{-}$. The cocycle is standard: 

    \begin{equation}
        c_{loop}(b_1, b_2) = \int_{S^1}\text{tr}(b_1 db_2), \ \ b_1, b_2\in Lb
    \end{equation}

\subsubsection{Diffeomorphisms of a circle}
    
To write down a momentum map for $\widetilde{\text{Diff}^+}(S^1)\curvearrowright\mathcal{O}_{\ell}^{L\bar{G}}$ we note that action $L\bar{G}\curvearrowright\mathcal{O}_{\ell}^{L\bar{G}}$ is transitive. Then for any $\varphi\in\widetilde{\text{Diff}^+}(S^1)$ there exists $R_{\varphi}$ such that $\varphi.(A, 1) = \text{Ad*}_{R_{\varphi}}(A, 1)$. We need the corresponding Lie algebra action, which might be written as follows: if $\varphi\mapsto f\partial\in\mathfrak{X}(S^1)$ and $R_{\varphi}\mapsto r_f\in L\bar{\g}$, then

    \begin{equation}\label{Mprime}
        \text{ad*}_{r_f}(A, 1) = f\partial.A \Longleftrightarrow [r_f, A] - r_f' = (fA)'
    \end{equation}
    Thus, we rewrote the action of $\mathfrak{X}(S^1)$ in terms of an action of $L\bar{\g}$. Fix an invarian bilinear form $\langle\cdot, \cdot\rangle$ on $\overline{\g}$ (do not mistake for a pairing between a functional and a vector, e.g. $\langle\mu_D(\ell), f\partial\rangle$).  Now let us compute the momentum map $\mu_D\colon\mathcal{O}_{\ell}^{L\bar{G}}\rightarrow\mathfrak{X}(S^1)^*$. Computation gives that 

    \begin{equation}
        \begin{split}
            d\langle\mu_D(\ell), f\partial\rangle &= \omega_{KKS}(\text{ad*}_{r_f}(A, 1), \bullet) = \langle(A, 1), [r_f, \bullet]_{\widehat{L\bar{\g}}}\rangle 
            = -\int_{S^1}\langle (fA)', \bullet\rangle d\theta
        \end{split}
    \end{equation}
    Thus, the momentum map is 

    \begin{equation}\label{DiffMM}
        \langle\mu_D(A, 1), f\partial\rangle = -\frac{1}{2}\int_{S^1}\text{tr}A^2f d\theta; \text{ i.e. } (A, 1) \mapsto -\frac{1}{2}\text{tr}A^2 (d\theta)^2
    \end{equation}

    Substituting this into (\ref{DefOfCocycle}), we obtain that the corresponding 2-cocycle vanishes:

    \begin{equation}\label{Diffcocycle}
        \begin{split}
            c(f\partial, g\partial) &= \mathcal{L}_{f\partial}\langle\mu_D(\ell), g\partial\rangle - \langle\mu_D(A, 1), [f\partial, g\partial]_{\mathfrak{X}(S^1)}\rangle =\\
            &= -\int_{S^1}\langle(fA)', r_g\rangle + \frac{1}{2}\int_{S^1}(fg' - f'g)\langle A, A\rangle d\theta = \\
            &= \int_{S^1}\langle fA, r_g'\rangle d\theta + \frac{1}{2}\int_{S^1}(fg' - f'g)\langle A, A\rangle d\theta = \textcolor{blue}{/\text{ identiry (\ref{Mprime})} /} =\\
            &= \int_{S^1}\langle fA, [r_g, A] - (gA)'\rangle d\theta + \frac{1}{2}\int_{S^1}(fg' - f'g)\langle A, A\rangle d\theta = \\
            &= -\int_{S^1}fg'\langle A, A\rangle d\theta + \frac{1}{2}\int_{S^1}(fg)'\langle A, A\rangle + \frac{1}{2}\int_{S^1}(fg' - f'g)\langle A, A\rangle d\theta = 0
        \end{split}
    \end{equation}
    Here we used integration by parts and the property $\langle A, [B, C]\rangle = \langle [A, B], C\rangle$. So, we see that in fact there is no cocycle for the action of $\widetilde{\text{Diff}^+}(S^1)$.

\subsubsection{Computation of the cocycle for the stabiliser}

    Now we want to apply our main Theorem \ref{MainTh} for $\mathcal{M} = \mathcal{O}_{\ell}^{L\bar{G}}$ and $G_{(\eta, a)}$, where $(\eta, a)$ is defined by (\ref{FuncEx}). To this end, we need to write down a momentum map $\mu$ for an action of $G = \widetilde{\text{Diff}^+}(S^1)\ltimes LB$ and then push it forward with an inclusion $i\colon\g_{(\eta, a)}\hookrightarrow\g$ to obtain a momentum map $\mu_{\eta}$ for an action of $G_{(\eta, a)}\leq G$. We can summarise this procedure by the following diagram: 

    \begin{center}
        \begin{tikzcd}
            \mathcal{O}_{\ell}^{L\bar{G}}\arrow[r, "\mu"]\arrow[rd, "\mu_{\eta}"] & \g^*\arrow[d, "i^*"]\\
             & \g^*_{(\eta, a)}
        \end{tikzcd}
    \end{center}

    The momentum map $\mu$ is given by combination of (\ref{LoopMM}) and (\ref{DiffMM}):

    \begin{equation}
        (A, 1) \mapsto \left(-\frac{1}{2}\text{tr}A^2 (d\theta)^2, A^{-}\right)
    \end{equation}

    Let $(f\partial, N_f), (g\partial, N_g)\in\g_{(\eta, a)}$. To compute a 2-cocycle we first compute the second term of (\ref{DefOfCocycle}):

    \begin{equation}\label{2term}
        \begin{split}
            -\langle\mu_{\eta}(A, 1), &[(f\partial, N_f), (g\partial, N_g)]_{\g_{(\eta, a)}}\rangle = \\
            &=-\left\langle\left(-\frac{1}{2}\text{tr}A^2 (d\theta)^2, A^{-}\right), ([f\partial, g\partial], [N_g, N_f] + fN_g' - gN_f')\right\rangle = \\
            &= \frac{1}{2}\int_{S^1}\text{tr}A^2(fg' - f'g)d\theta - \int_{S^1}\langle A^{-},[N_g, N_f] \rangle d\theta - \int_{S^1}f\langle A^{-},N_g' \rangle d\theta + \int_{S^1}g\langle A^{-},N_f' \rangle d\theta
        \end{split}
    \end{equation}

    The first term in (\ref{DefOfCocycle}) is (terms of the same colour cancel out)

\begin{equation}\label{1term}
        \begin{split}
            \{\ldots, \ldots\} &= \omega_{KKS}(\text{ad*}_{(f\partial, N_f)}(A, 1), \text{ad*}_{(g\partial, N_g)}(A, 1))  = \langle (A, 1), [r_f + N_f, r_g + N_g]\rangle = \\
            &= \int_{S^1}\langle A, [r_f, r_g]\rangle d\theta + \int_{S^1}\langle A, [r_f, N_g]\rangle d\theta + \int_{S^1}\langle A, [N_f, r_g]\rangle d\theta + \int_{S^1}\langle A, [N_f, N_g]\rangle d\theta + \\
            &+ c_1(r_f, r_g) + c_1(r_f, N_g) + c_1(N_f, r_g) + c_1(N_f, N_g) = \textcolor{blue}{/\text{tr}A[B, C] = \text{tr}[A,B]C/} \\
            &= \int_{S^1}\langle [A, r_f], r_g\rangle d\theta + \int_{S^1}\langle [A, r_f], N_g\rangle d\theta - \int_{S^1}\langle [A, r_g], N_f\rangle d\theta + \int_{S^1}\langle A, [N_f, N_g]\rangle d\theta + \\
            &+ c_1(r_f, r_g) + c_1(r_f, N_g) - c_1(r_g, N_f) + c_1(N_f, N_g) = \textcolor{blue}{/\text{identity (\ref{Mprime})/}}\\
            &= \int_{S^1}\langle \textcolor{red}{-r_f'} - (fA)', \textcolor{red}{r_g}\rangle d\theta + \int_{S^1}\langle \textcolor{green}{-r_f'} - (fA)', \textcolor{green}{N_g}\rangle d\theta - \int_{S^1}\langle \textcolor{orange}{-r_g'} - (gA)', \textcolor{orange}{N_f}\rangle d\theta + \\
            &+ \int_{S^1}\langle A, [N_f, N_g]\rangle d\theta + \\
            &+ \int_{S^1}\textcolor{red}{\langle r_f', r_g\rangle} d\theta + \int_{S^1}\textcolor{green}{\langle r_f', N_g\rangle} d\theta - \int_{S^1}\textcolor{orange}{\langle r_g', N_f\rangle} d\theta + c_1(N_f, N_g) = \\
            &= -\frac{1}{2}\int_{S^1}\text{tr}A^2(fg' - f'g)d\theta + \int_{S^1}\langle A,[N_g, N_f] \rangle d\theta + \int_{S^1}f\langle A,N_g' \rangle d\theta - \int_{S^1}g\langle A,N_f' \rangle d\theta + \\
            &+ c_1(N_f, N_g)
        \end{split}
    \end{equation}
    Here we also used integration by parts. Now we see that the sum of (\ref{1term}) and (\ref{2term}) is

    \begin{equation}\label{FinalCoc}
        \begin{split}
            c_{(\eta, a)}((f\partial, N_f), (g\partial, N_g)) &= \int_{S^1}\langle A^{+},[N_g, N_f] \rangle d\theta + \int_{S^1}f\langle A^{+},N_g' \rangle d\theta - \int_{S^1}g\langle A^{+},N_f' \rangle d\theta +\\
            &+ c_1(N_f, N_g) = \\
            &= 0 + 0 + 0 + \frac{1}{4}\int_{S^1}\text{tr}\begin{pmatrix}
                -f'' & \tilde{\lambda}f'''\\
                0 & f''
            \end{pmatrix}
            \begin{pmatrix}
                -g'' & \tilde{\lambda}g'''\\
                0 & g''
            \end{pmatrix}d\theta = \\ 
            &= \frac{1}{2}\int_{S^1}f''g' d\theta = c_3(f\partial, g\partial)
        \end{split}
    \end{equation}

    So, the desired cocycle is the Gelfand-Fuchs 2-cocycle (up to a non-zero constant). The only other thing that could happen is that the cocycle might be zero as in (\ref{Diffcocycle}). Here, we want to emphasise the following. In the case of an action of $\widetilde{\text{Diff}^+}(S^1)$, whose Lie algebra $\mathfrak{X}(S^1)$ {\it does} have a central extension $\mathfrak{vir}$ given by $c_3(\cdot, \cdot)$, there is no 2-cocycle -- see (\ref{Diffcocycle}). And in the case of the stabiliser $G_{(\eta, a)}$, whose Lie algebra is isomorphic to $\mathfrak{X}(S^1)$ and it {\it does not} have the central charge (see the beginning of the section), there is a 2-cocycle -- namely, the Gelfand-Fuchs cocycle. In other words, the cocycle appears either in Lie algebra or in a momentum map, and it depends on the embedding $\mathfrak{X}(S^1)\hookrightarrow\mathfrak{X}(S^1)\ltimes Lb$.

\paragraph{Remark.} If we substitute our functional $\eta = \begin{pmatrix}
    0 & 0\\
    1 & 0
\end{pmatrix}\in\h^*$ to $\mu_{loop}^{-1}(\eta)$, then we obtain 

\begin{equation}
    \mathcal{M}' = (\mathcal{O}^{L\overline{G}}_{\ell}\times\mathcal{O}^H_{(\eta, a)})\sslash_0 H = \left\{\begin{pmatrix}
        0 & T\\
        1 & 0
    \end{pmatrix}dx \ \ \middle| \ \ T\in C^{\infty}(S^1)\right\}.
\end{equation}
Moreover, there is an action of $\widetilde{\text{Diff}^+}(S^1)$ with the Gelfand-Fuchs 2-cocycle. In other words, the set $\mathcal{M}'$ is nothing but the space of Hill operators. This space might be regarded as a hyperplane in $\mathfrak{vir}^*$ (see \cite{GeomInfDimGr} Part II, Chapter 2).

We think we might consider this observation as a version of Drinfeld-Sokolov reduction \cite{DrinSok}. Let us recall its key steps (for a concise presentation of the subject we refer to Appendix 8 in \cite{GeomInfDimGr}): 

\begin{itemize}

    \item There is a coadjoint action $L\text{GL}_2 \curvearrowright\{-a\partial + A \ \ \vert\ \ A\in  L\mathfrak{gl}_2, a\in\R\}\subseteq\widehat{L\mathfrak{gl}_2^*}$ (one could put here an arbitrary positive integer $n$ instead of 2 but we only need $n = 2$). Hyperplanes $\{a = \text{const}\}$ are invariant under this action, so we fix $a = -1$. 

    \item The loop group $LN_{-}$ of lower triangular matrices with 1's on the diagonal still acts on the hyperplane $H_{-1} = \{\partial + A \ \ \vert \ \ A \in L\mathfrak{gl}_2\}$ by the coadjoint action. The action is Hamiltonian with the momentum map $\Phi\colon H_{-1}\longrightarrow L\mathfrak{n}_{+}$ given by the natural projection. An element $\Lambda = \begin{pmatrix}
        0 & -1 \\
        0 & 0
    \end{pmatrix}\in L\mathfrak{n}_{+}$ is fixed under the (conjugation) action of the group $LN_{-}$. 
    
    \item The hyperplane $H_{-1}$ carries a natural Poisson structure inherited from $\widehat{L\mathfrak{gl}_2}$. This Poisson structure gives some Poisson structure on the symplectic quotient $\Phi^{-1}(\Lambda)/LN_{-}$ as a result of Hamiltonian reduction. 

    \item The result of Drinfeld and Sokolov is that this Poisson structure on $\Phi^{-1}(\Lambda)/LN_{-}$ coincides with the quadratic Gelfand-Dickey structure on the space $\mathcal{L}_2$ on smooth monic 2nd-order differential operator on the circle. This space can be regarded as the space of Hill's operators that we discussed above.
\end{itemize}

We see that our procedure resembles the Drinfeld-Sokolov reduction. The main difference is that we use the Borel subgroup for the reduction instead of a nilpotent one. We believe one can say more about the relation between these two reductions.

\subsection{The second example: character variety}

Although the following example of $\mathcal{M}$ is of great interest itself, for us, it serves as a preparation for the third example. 

Fix a group $\overline{G} = \text{PSL}_2(\R)$ and its Lie algebra $\overline{\g} = \mathfrak{sl}_2(\R)$ (we use bar just to distinguish these objects from $G$ and $\g$). According to Riemann-Hilbert correspondence the character variety $\text{Rep}(\pi_1(M), \overline{G})$ describes the moduli space of flat connections on a manifold $M$, i.e.

\begin{equation}\label{RGcorr}
    \text{Rep}(\pi_1(M), \overline{G}) \cong \{flat \ \ \overline{G}-connections \ \ on \ \ M\}/gauge
\end{equation}

We will not go into details of this theorem and character varieties (they might be found, for example, in \cite{GentInt}) and will focus on the right-hand side of (\ref{RGcorr}). 

As $M$ we take a two-dimensional compact surface $\Sigma$ with one boundary component $\partial\Sigma=S^1$(for the boundary-free case see \cite{GentInt}, and \cite{Goldman} for relevant questions). We put $\mathcal{A} = \{d + A\vert A\in\Omega^1(\Sigma, \mathfrak{sl}_2(\R))\}$ -- the space of all connections of the trivial $\overline{G}$-bundle over $\Sigma$. This space possesses a symplectic form known as Atiyah-Bott form: 

\begin{equation}
    \omega_{AB}(A_1, A_2) = \int_{\Sigma}\text{tr}A_1\wedge A_2
\end{equation}

The gauge group will be $\mathcal{G} = \{g\colon\Sigma\rightarrow\overline{G}\}$ - the group of all $\overline{G}$-valued maps from $\Sigma$. This group has a normal subgroup of maps, trivial on the boundary: $\mathcal{G}_0 = \{g\colon\Sigma\rightarrow\overline{G}: g\vert_{\partial\Sigma} = e\}$. Now, let us construct our $\mathcal{M}$.  

\vspace{10pt}

First, we note that there is a Hamiltonian action $\mathcal{G}\curvearrowright\mathcal{A}$ given by gauge transformations: $g.A = gAg^{-1} - dgg^{-1}$. The corresponding momentum map is the curvature and the restriction $A_{\theta} = A\vert_{\partial\Sigma}$ of $A$ to the boundary (see \cite{GeomInfDimGr} for details): 

\begin{equation}
    \begin{split}
        \mu\colon\mathcal{A}&\longrightarrow\text{Lie}(\mathcal{G})^* \cong\Omega^2(\Sigma, \mathfrak{sl}_2(\R))\times\Omega^1(\partial\Sigma, \overline{\g})\\
        A &\longmapsto (F_A = dA + A\wedge A, A_{\theta}) \\
        \langle\mu(A), &x\rangle = \int_{S^1}\text{tr}(xF_A) - \int_{\partial\Sigma} \text{tr}(xA_{\theta}), \forall x\in\text{Lie}(\mathcal{G})
    \end{split}
\end{equation}
If $\Sigma$ did not have a boundary, there would be no cocycle. But in the case of $\partial\Sigma=S^1$ we have

\begin{equation}\label{CocycleCon}
    c_{con}(x, y) = \int_{\partial\Sigma}\text{tr}(xdy),
\end{equation}
where $x, y\in\Omega^0(\Sigma, \mathfrak{sl}_2(\R)) = \text{Lie}(\mathcal{G})$. 
If we restrict the action to the subgroup $\mathcal{G}_0$, then $x$ and $y$ from $\text{Lie}(\mathcal{G}_0)$ are zeros on the boundary, so the cocycle $c_{con}$ vanishes. Hence, according to the symplectic reduction procedure, we can define a symplectic quotient $\mathcal{A}\sslash_0\mathcal{G}_0$. This is our $\mathcal{M}$: 

\begin{equation}\label{SecExM}
    \mathcal{M} := \mathcal{A}\sslash_0\mathcal{G}_0 = \frac{\{d+A \ \  | \ \ A\in\Omega^1(\Sigma, \overline{\g})
    \big| F_A = 0\}}{\{g\colon\Sigma\rightarrow\overline{G}:g\vert_{\partial\Sigma} = e\}} =: \widetilde{\mathcal{M}}/\mathcal{G}_0
\end{equation}
This is not yet a charter variety, but close to it: one needs to carry out a reduction w.r.t. $\mathcal{G}/\mathcal{G}_0$ (see \cite{GentInt}). Now, when the symplectic manifold $\mathcal{M}$ is constructed, we can proceed to our subject (details on the construction of the space and its further study might be found in \cite{MeinWood}).
\subsubsection{Loop group}

Since there is an action $\mathcal{G}\curvearrowright\mathcal{A}$ and $\mathcal{M}$ is obtained by reduction w.r.t. normal subgroup $\mathcal{G}_0\trianglelefteq\mathcal{G}$, there is a residual action of $\mathcal{G}/\mathcal{G}_0\cong L\overline{G}$ on $\mathcal{M}$. Again, we restrict this action to a subgroup $LB\leq L\overline{G}$ and get the corresponding momentum map:

\begin{equation}\label{CVloopMM}
    \begin{split}
        \mu_{loop}\colon\mathcal{M}&\longrightarrow L\overline{\g}^* \longrightarrow Lb^*\cong Lb^{-}\\
        [A]&\longmapsto A^{\theta}\longmapsto A^{-}_{\theta}\\
        \langle\mu_{loop}(&[A], b)\rangle = \int_{\partial\Sigma}\text{tr}A^{-}_{\theta}b
    \end{split}
\end{equation}
And the cocycle (\ref{CocycleCon}) becomes

\begin{equation}
    c_{loop}(b_1, b_2) = \int_{\partial\Sigma}\text{tr}(b_1 db_2)
\end{equation}

\subsubsection{Diffeomorphisms of a circle}

One can consider an action on $\widetilde{\mathcal{M}}$ by the group $\text{Diff}_{\circ}(\Sigma)$ of diffeomorphisms of $\Sigma$ isotopic to identity. This action is basically a change of coordinates on the surface and it descends to $\mathcal{M}$. Now, note that normal subgroup $\leftindex^0 {\text{Diff}_{\circ}(\Sigma)}\trianglelefteq\text{Diff}_{\circ}(\Sigma)$ fixing the boundary acts trivially on $\mathcal{M}$. Hence, there is a residual action on $\mathcal{M}$ by $\widetilde{\text{Diff}^+}(\partial\Sigma) = \text{Diff}_{\circ}(\Sigma)/\leftindex^0 {\text{Diff}_{\circ}(\Sigma)}$: 

\begin{equation}
    \varphi.[A] = [\Phi.(\Phi_0.A)],
\end{equation}
where $\Phi\in\text{Diff}_{\circ}(\Sigma)$, $\Phi_0\in\leftindex^0 {\text{Diff}_{\circ}(\Sigma)}$. We get the corresponding momentum map $\mu_D$ using the momentum map (\ref{CVloopMM}) for the loop group. The action of $f\partial\in\mathfrak{X}(S^1)$ on $A_{\theta}$ might be rewritten in terms of an action of $X_f = \iota_{f\partial} A_{\theta}\in L\bar{\g}$ in the following way (compare it to (\ref{Mprime})): 

\begin{equation}
    \begin{split}
        \mathcal{L}_{f\partial}A_{\theta} &= d(\iota_{f\partial} A_{\theta}) + \iota_{f\partial}dA_{\theta} = \textcolor{blue}{/dA + [A, A] = 0/}\\
        &= dX_f - \frac{1}{2}\iota_{f\partial}[A, A] = dX_f - [X_f, A_{\theta}] = \\
        &= \text{ad*}_{X_f}A_{\theta}
    \end{split}
\end{equation}
So, the momentum map $\mu_D$ now reads as

\begin{equation}
    \begin{split}
        \langle\mu_D(A), f\partial\rangle &= \langle\mu_{loop}(A), X_f\rangle = \int_{\partial\Sigma}\text{tr}(dX_f - [X_f, A_{\theta}]) =\\
        &= 0 - \frac{1}{2}\int_{\partial\Sigma}\text{tr}(\iota_{f\partial}A_{\theta}\wedge A_{\theta}) = -\frac{1}{2}\int_{\partial\Sigma}\text{tr}fA_{\theta}^2
    \end{split}
\end{equation}
So, we have

\begin{equation}\label{CVDiffMM}
    \begin{split}
    \mu_D\colon\mathcal{M}&\longrightarrow\mathfrak{X}(S^1)^*\\
    [A]&\longmapsto -\frac{1}{2}\text{tr}A_{\theta}^2
    \end{split}
\end{equation}
\subsubsection{Stabiliser}

Compare momentum map (\ref{LoopMM}) with (\ref{CVloopMM}) and (\ref{DiffMM}) with (\ref{CVDiffMM}). We see that they are basically the same. Then, the cocycle for the stabiliser $G_{(\eta, a)}$ is the Gelfand-Fuchs 2-cocycle as in (\ref{FinalCoc}).

\subsection{The third example: Teichmüller space}

The following example is the main motivation for the present work. Recall that it is taken from the paper \cite{AlArt} and we refer to it for the details. The example is similar to the previous one. The difference is that the space $\widetilde{\mathcal{M}}$ in (\ref{SecExM}) is "smaller" and the corresponding gauge group is also "smaller". Let us be more precise. 

\begin{itemize}
    \item As before we start with a compact two-dimensional surface $\Sigma$ with one boundary component $\partial\Sigma = S^1$. Then we consider a space $\text{Hyp}(\Sigma)$ of all hyperbolic structures on $\Sigma$. A hyperbolic structure on $\Sigma$ is an oriented atlas with $\overline{\mathbb{D}}$-valued charts, with constant transition maps given by elements of $G = \text{PSU}(1,1)$. One can see a hyperbolic structure as a hyperbolic 0-metric $\mathsf{g}$, where "0" refers to a certain behaviour of a metric near the boundary (namely, as for the standard hyperbolic metric on the Poincaré disk). The Teichmüller space is the quotient

    \begin{equation}
        \text{Teich}(\Sigma) = \text{Hyp}(\Sigma)/\leftindex^0 {\text{Diff}_{\circ}(\Sigma)},
    \end{equation}
    where $\leftindex^0 {\text{Diff}_{\circ}(\Sigma)}$ is the group of diffemorphisms of $\Sigma$ isotopic to identity and fixing the boundary. 

    \item Using the coframe formalism (Section 3 in \cite{AlArt}), for a given 0-metric $\mathsf{g} = (\alpha_1)^2 + (\alpha_2)^2$ one can construct a $\mathfrak{sl}_2(\R)$-valued connection 1-form $A = \frac{1}{2}\begin{pmatrix}
        \alpha_2 & \alpha_1 - \kappa\\
        \alpha_1 + \kappa & -\alpha_2
    \end{pmatrix}$, where $\alpha_1, \alpha_2, \kappa\in\leftindex^0{\Omega^1(U, \mathfrak{sl}_2(\R))}$ for an open subset $U\subseteq\Sigma$. Then, one can show that a 0-metric $\mathsf{g}$ is hyperbolic if and only if an associated connection $A$ is flat, i.e. $F_A = dA + \frac{1}{2}[A, A] = 0$. This is how we pass from hyperbolic 0-metrics to flat connections. 

    \item Then we consider a space 

    \begin{equation}
        \widehat{\text{Teich}}(\Sigma) = \mathcal{A}_{flat}^{pos}(P)/\text{Gauge},
    \end{equation}
    where $\mathcal{A}_{flat}^{pos}(P)$ is the space of flat positive connections of the principal $G$-bundle $P$ over $\Sigma$. This is an analogue of $\widetilde{\mathcal{M}}$ from the previous example: we still have a space of flat connections, but now they also satisfy certain positivity condition, which we do not discuss here. The corresponding gauge group is also smaller since it must preserve this positivity.

    \item Next, Theorem 6.7 in \cite{AlArt} tells us that the Teichmüller space is a symplectic quotient 

    \begin{equation}
        \text{Teich}(\Sigma) = (\widehat{\text{Teich}}(\Sigma)\times\mathcal{O^{-}})\sslash_0\text{Gau}(Q, \tau),
    \end{equation}
    where $\text{Gau}(Q, \tau)$ is our subgroup $H = LB \trianglelefteq \widetilde{\text{Diff}}^+(S^1) \ltimes L B = G$. So we see that this is exactly the setup of our main theorem for $\mathcal{M} = \widehat{\text{Teich}}(\Sigma)$ and $\mathcal{M}' = \text{Teich}(\Sigma)$.
\end{itemize}

\subsubsection{The momemtum map and the cocycle}

As before, to compute a 2-cocycle for a stabiliser $G_{(\eta, a)}\cong\widetilde{\text{Diff}}^+(S^1)$ we need to know a momentum map for the group $G = \widetilde{\text{Diff}}^+(S^1) \ltimes L B$. The map is computed in the paper \cite{AlArt} (see Proposition 7.7): 

\begin{equation}
    \begin{split}
        \mu\colon \widehat{\text{Teich}}(\Sigma)&\longrightarrow\g^*\\
        [\theta]&\longmapsto\left(\frac{1}{2}\text{tr}A^2, A\right),
    \end{split}
\end{equation}
where $A$ is a trivialisation of the connection $\theta$. We see that the momentum map is exactly the same as in previous examples (the only difference is the sign of the first component, but this does not influence computations). Thus, we can conclude that the desired 2-cocycle for the action of the stabiliser will be the Gelfand-Fuchs 2-cocycle.

\bibliographystyle{abbrv}

\end{document}